\documentclass[runningheads,envcountsame,a4paper]{llncs}

\usepackage{amsfonts,amsmath,amssymb}
\usepackage{graphicx}
\usepackage{clrscode}
\usepackage{bbold}
\usepackage[hyphens]{url}
\usepackage{microtype}
\usepackage{tikz}
\usetikzlibrary{topaths,calc,automata,positioning}
\usepackage{booktabs}

\DeclareMathOperator*{\im}{im}
\newcommand{\seepage}[1]{\marginpar{\scriptsize (p.~\pageref{#1})}}

\begin{document}
\title{Descriptional Complexity  of Semi-Simple Splicing Systems}
\author{Lila Kari \and Timothy Ng}

\institute{School of Computer Science, University of Waterloo\\
Waterloo, Ontario N2L 3G1, Canada\\
  \email{\{lila.kari, tim.ng\}@uwaterloo.ca}
}

\maketitle

\begin{abstract}
Splicing systems are generative mechanisms introduced by Tom Head in 1987 to model the biological process of DNA recombination. The computational engine of a splicing system is the ``splicing operation",  a cut-and-paste binary string operation defined by a set  of ``splicing rules" $r = (\alpha_1, \alpha_2 ; \alpha_3, \alpha_4)$ where $\alpha_1, \alpha_2, \alpha_3, \alpha_4$  are words over an alphabet $\Sigma$. For two strings $x = x_1 \alpha_1 \alpha_2 x_2$ and $y = y_1 \alpha_3 \alpha_4 y_2$, applying the  splicing rule $r$ produces the string $z = x_1 \alpha_1 \alpha_4 y_2$. 
In this paper we  focus on a particular type of  splicing systems, called $(i, j)$ semi-simple splicing systems, $i = 1,2$ and $j = 3, 4$, wherein all splicing rules  have the property that the  two strings in positions $i$ and $j$ are   singleton letters, while the other two strings are empty. The language    generated by such a system consists of the set of words that are obtained starting from an initial set called ``axiom set", by iteratively applying the splicing rules to  strings in the axiom set as well as to intermediately produced strings. We consider semi-simple splicing systems where the axiom set is a regular language, and investigate the descriptional complexity of such systems  in terms of the size of the minimal deterministic finite automata that recognize the languages they generate.  

\end{abstract}

\section{Introduction}

Splicing systems are generative mechanisms introduced by Tom Head \cite{Head1987} to model the biological process of DNA recombination. A splicing system consists of an initial language called an {\it axiom set}, and a set of so-called  {\it splicing rules}. The result of applying a splicing rule  to a pair of operand strings  is a new ``recombinant" string, and  the language generated by a splicing system consists  all the words that can be obtained by successively applying splicing rules to axioms and the intermediately produced words. The most natural variant of splicing systems, often referred to as {\em finite splicing systems}, is to consider a finite set of axioms and a finite set of rules. 
Several different types of splicing systems  have been proposed in the literature,  and Bonizzoni et al.~\cite{Bonizzoni2001} showed that the classes of languages they generate are related: the class of languages generated by finite Head splicing systems \cite{Head1987} is strictly contained in the class of languages generated by finite P\u{a}un splicing systems \cite{Paun1996}, which is strictly contained in the class of languages generated by finite Pixton splicing systems \cite{Pixton1996}.


In this paper we will use the P\u{a}un definition \cite{Paun1996}, which defines a splicing rule   
as a quadruplet of words $r = (u_1,v_1;u_2,v_2)$.
This rule splices two words 
$x_1u_1v_1y_1$ and $x_2u_2v_2y_2$ as follows:
The words are cut between the factors $u_1,v_1$, respectively $u_2,v_2$, and
the prefix  of the first word (ending in $u_1$) is 
recombined by catenation with the suffix  of the second word (starting with $v_2$), resulting in the word $x_1 u_1 v_2 y_2$.
 
 Culik~II and Harju \cite{Culik1991} proved that
finite Head splicing systems can only generate regular languages, while also \cite{Head2006} and \cite{Pixton1996}  proved a similar  result for P\u{a}un, respectively Pixton splicing systems.
Gatterdam \cite{Gatterdam1989} gave $(aa)^*$ as an example of a
regular language which cannot be generated by a finite Head splicing system, which proved that this is a strict inclusion.


As the classes of languages generated by finite splicing systems are subclasses of the family of regular languages, their descriptional complexity can be considered in terms of the finite automata that recognize them.
For example, Loos et al.~\cite{Loos2008} gave a bound on the number of states required for a nondeterministic finite automaton to recognize the language generated by an equivalent Paun finite splicing system.
Other descriptional complexity measures  that have been investigated in the literature include  the number of rules, the number of words in the initial language, the maximum length of a word in the initial axiom set, and the sum of the lengths of all words in the axiom set, for simple  splicing systems (see below),  ~\cite{Mateescu1998}; the number of rules, the length of the rules, and the size of the axiom set;  the radius, the largest $u_i$ in a rule, P\u{a}un~\cite{Paun1996}.

In the original definition, simple splicing systems are finite  splicing systems where all the words in the splicing rules are singleton letters, and  the descriptional complexity of simple splicing systems was considered by Mateescu et al.~\cite{Mateescu1998}  in terms of the size of a right linear grammar that generates a simple splicing language. 
Semi-simple splicing systems were introduced in Goode and Pixton \cite{Goode2001} as having a finite  axiom set, and splicing rules  of the form $(a, \epsilon; b, \epsilon)$ where $a, b$ are singleton letters, and $\epsilon$ denotes the empty word. 

In this paper we  consider an expanded  definition of semi-simple splicing systems that allows the axiom set to  be a regular language. We focus  our study on some variants of  such semi-simple splicing systems, called $(i, j)$ semi-simple splicing systems, $i = 1,2$ and $j = 3, 4$, wherein all splicing rules  have the property that the  two strings in positions $i$ and $j$ are   singleton letters, while the other two strings are empty. Ceterchi et al. \cite{Ceterchi2003} showed that  all classes of languages generated by semi-simple splicing systems are pairwise incomparable\footnote{Simple splicing language classes are pairwise incomparable except for  the pair (1,3) and (2,4), which are equivalent~\cite{Mateescu1998}}. In addition, in a departure from the  original definition of semi-simple splicing systems \cite{Goode2001}, in this paper the  axiom set is allowed to be a (potentially infinite) regular set. 

More precisely,  we investigate the descriptional complexity of   $(i, j)$ semi-simple splicing systems with regular axiom sets, in terms of the size of the minimal deterministic finite automaton that recognizes the language generated by the system.
The paper is organized as follows:  Section \ref{sec:prelim}  introduces definitions and notations, Section \ref{sec:splicing} defines  splicing systems and outlines some basic results on simple splicing systems, Sections \ref{sec:24}, \ref{sec:23}, \ref{sec:14} investigate the state complexity of (2,4), (2,3) respectively (1,4) semi-simple splicing systems, and Section \ref{sec:concl} summarizes our results (Table \ref{tab:conclusion}).

\section{Preliminaries}
\label{sec:prelim}

Let $\Sigma$ be a finite alphabet. We denote by $\Sigma^*$ the set of all finite words over $\Sigma$, including the empty word, which we denote by $\varepsilon$. We denote the length of a word $w$ by $|w| = n$. If $w = xyz$ for $x,y,z \in \Sigma^*$, we say that $x$ is a prefix of $w$, $y$ is a factor of $w$, and $z$ is a suffix of $w$.

A deterministic finite automaton (DFA) is a tuple $A = (Q,\Sigma,\delta,q_0,F)$ where $Q$ is a finite set of states, $\Sigma$ is an alphabet, $\delta$ is a function $\delta : Q \times \Sigma \to Q$, $s \in Q$ is the initial state, and $F \subseteq Q$ is a set of final states. We extend the transition function $\delta$ to a function $Q \times \Sigma^* \to Q$ in the usual way. A DFA $A$ is complete if $\delta$ is defined for all $q \in Q$ and $a \in \Sigma$. In this paper, all DFAs are defined to be complete. We will also make use of the notation $q \xrightarrow{w} q'$ for $\delta(q,w) = q'$, where $w \in \Sigma^*$ and $q,q' \in Q$. The language recognized or accepted by $A$ is $L(A) = \{w \in \Sigma^* \mid \delta(q_0,w) \in F\}$.

Each letter $a \in \Sigma$ defines a transformation of the state set $Q$. Let $\delta_a : Q \to Q$ be the transformation on $Q$ induced by $a$, defined by $\delta_a(q) = \delta(q,a)$. We extend this definition to words by composing the transformations $\delta_w = \delta_{a_1} \circ \delta_{a_2} \circ \cdots \circ \delta_{a_n}$ for $w = a_1 a_2 \cdots a_n$. 
We denote by $\im \delta_a$ the image of $\delta_a$, defined $\im \delta_a = \{ \delta(p,a) \mid p \in Q \}$.

A state $q$ is called \emph{reachable} if there exists a string $w \in \Sigma^*$ such that $\delta(q_0,w) = q$. A state $q$ is called \emph{useful} if there exists a string $w \in \Sigma^*$ such that $\delta(q,w) \in F$. A state that is not useful is called \emph{useless}. A complete DFA with multiple useless states can be easily transformed into an equivalent DFA with at most one useless state, which we refer to as the \emph{sink state}.

Two states $p$ and $q$ of $A$ are said to be \emph{equivalent} in the case that $\delta(p,w) \in F$ if and only if $\delta(q,w) \in F$ for every word $w \in \Sigma^*$. A DFA $A$ is minimal if each state $q \in Q$ is reachable from the initial state and no two states are equivalent. The state complexity of a regular language $L$ is the number of states of the minimal complete DFA recognizing $L$ \cite{Gao2016}.

A nondeterministic finite automaton (NFA) is a tuple $A = (Q,\Sigma,\delta,I,F)$ where $Q$ is a finite set of states, $\Sigma$ is an alphabet, $\delta$ is a function $\delta: Q \times \Sigma \to 2^Q$, $I \subseteq Q$ is a set of initial states, and $F \subseteq Q$ is a set of final states. The language recognized by an NFA $A$ is $L(A) = \{ w \in \Sigma^* \mid \bigcup_{q \in I} \delta(q,w) \cap F \neq \emptyset \}$. As with DFAs, transitions of $A$ can be viewed as transformations on the state set. Let $\delta_a : Q \to 2^Q$ be the transformation on $Q$ induced by $a$, defined by $\delta_a(q) = \delta(q,a)$. We define $\im \delta_a = \bigcup_{q \in Q} \delta_a(q)$. We make use of the notation $P \xrightarrow{w} P'$ for $P' = \bigcup_{q \in P} \delta(q,w)$, where $w \in \Sigma^*$ and $P,P' \subseteq Q$. 


\section{Semi-simple Splicing Systems} \label{sec:splicing}


In this paper we will use the notation of P\u{a}un~\cite{Paun1996}. The splicing operation is defined via sets of quadruples $r = (\alpha_1, \alpha_2 ; \alpha_3, \alpha_4)$ with $\alpha_1, \alpha_2, \alpha_3, \alpha_4 \in \Sigma^*$ called splicing rules. For two strings $x = x_1 \alpha_1 \alpha_2 x_2$ and $y = y_1 \alpha_3 \alpha_4 y_2$, applying the rule $r = (\alpha_1, \alpha_2 ; \alpha_3, \alpha_4)$ produces a string $z = x_1 \alpha_1 \alpha_4 y_2$, which we denote by $(x,y) \vdash^r z$. 

A \emph{splicing scheme} is a pair $\sigma = (\Sigma,\mathcal R)$ where $\Sigma$ is an alphabet and $\mathcal R$ is a set of splicing rules. For a splicing scheme $\sigma = (\Sigma,\mathcal R)$ and a language $L \subseteq \Sigma^*$, we denote by $\sigma(L)$ the language
\begin{equation*}
\sigma(L) = L \cup \{z \in \Sigma^* \mid \text{$(x,y) \vdash^r z$, where $x,y \in L, r \in \mathcal R$} \}.
\end{equation*}
Then we define $\sigma^0(L) = L$ and $\sigma^{i+1}(L) = \sigma(\sigma^i(L))$ for $i \geq 0$ and
\begin{equation*}
\sigma^*(L) = \lim_{i \to \infty} \sigma^i(L) = \bigcup_{i \geq 0} \sigma^i(L).
\end{equation*}
For a splicing scheme $\sigma = (\Sigma, \mathcal R)$ and an initial language $L \subseteq \Sigma^*$, we say the triple $H = (\Sigma, \mathcal R, L)$ is a \emph{splicing system}. The language generated by $H$ is defined by $L(H) = \sigma^*(L)$.

Goode and Pixton \cite{Goode2001} define a restricted class of splicing systems called semi-simple splicing systems. A semi-simple splicing system is a triple $H = (\Sigma,M,I)$, where $\Sigma$ is an alphabet, $M \subseteq \Sigma \times \Sigma$ is a set of markers, and $I$ is a finite initial language over $\Sigma$. We have $(x,y) \vdash^{(a,b)} z$ if and only if $x = x_1 a x_2$, $y = y_1 b y_2$, and $z = x_1 a y_2$ for some $x_1, x_2, y_1, y_2 \in \Sigma^*$. That is, a semi-simple splicing system is a splicing system in which the set of rules is $\mathcal M = \{(a,\varepsilon;b,\varepsilon) \mid (a,b) \in M \}$. Since the rules are determined solely by our choice of $M \subseteq \Sigma \times \Sigma$, the set $M$ is used in the definition of the semi-simple splicing system rather than the set of rules $\mathcal M$. 

It is shown in~\cite{Goode2001} that the class of languages generated by semi-simple splicing systems is a subclass of the regular languages. Semi-simple splicing systems are a generalization of the class of simple splicing systems, defined by Mateescu et al.~\cite{Mateescu1998}. A splicing system is a simple splicing system if it is a semi-simple splicing system and all markers are of the form $(a,a)$ for $a \in \Sigma$. It is shown in~\cite{Mateescu1998} that the class of languages generated by simple splicing systems is a subclass of the extended star-free languages.

Observe that the set of rules $\mathcal M = \{(a,\varepsilon;b,\varepsilon) \mid (a,b) \in M\}$ of a semi-simple splicing system consist of 4-tuples with symbols from $\Sigma$ in positions 1 and 3 and $\varepsilon$ in positions 2 and 4. We can call such splicing rules (1,3)-splicing rules. Then a (1,3)-splicing system is a splicing system with only (1,3)-splicing rules and ordinary semi-simple splicing systems can be considered (1,3)-semi-simple splicing systems. The state complexity of (1,3)-simple and (1,3)-semi-simple splicing systems was studied previously by the authors in \cite{Kari2019a}.

We can consider variants of semi-simple splicing systems in this way by defining semi-simple $(i,j)$-splicing systems, for $i = 1,2$ and $j = 3,4$. A semi-simple (2,4)-splicing system is a splicing system $(\Sigma,M,I)$ with rules $\mathcal M = \{(\varepsilon,a;\varepsilon,b) \mid (a,b) \in M\}$. A (2,3)-semi-simple splicing system is a splicing system $(\Sigma,M,I)$ with rules $\mathcal M = \{(\varepsilon,a;b,\varepsilon) \mid (a,b) \in M\}$. A (1,4)-semi-simple splicing system is a semi-simple splicing system $(\Sigma,M,I)$ with rules $\mathcal M = \{(a,\varepsilon;\varepsilon,b) \mid (a,b) \in M\}$. 

Ceterchi et al.~\cite{Ceterchi2003} show that class of languages generated by (1,3)-, (1,4)-, (2,3)-, and (2,4)-semi-simple splicing systems are all incomparable. However, it was shown in~\cite{Mateescu1998} that the classes of languages generated by (1,3)-simple splicing systems (i.e. ordinary simple splicing systems) and (2,4)-simple splicing systems are equivalent, while, the classes of languages generated by (1,3)-, (1,4)-, and (2,3)-simple splicing systems are all incomparable and subregular. 

In this paper, we will relax the condition that the initial language of a semi-simple splicing system must be a finite language. We will consider also semi-simple splicing systems with regular initial languages. By \cite{Paun1996}, it is clear that such a splicing system will also produce a regular language. In the following, we will use the convention that $I$ denotes a finite language and $L$ denotes an infinite language. 

\section{State Complexity of (2,4)-semi-simple Splicing Systems} \label{sec:24}
In this section, we will consider the state complexity of (2,4)-semi-simple splicing systems. Recall that a (2,4)-semi-simple splicing system is a splicing system with rules of the form $(\varepsilon,a;\varepsilon,b)$ for $a, b \in \Sigma$. As mentioned previously, the classes of languages generated by (1,3)- and (2,4)-simple splicing systems were shown to be equivalent by Mateescu et al.~\cite{Mateescu1998}, while the classes of languages generated by (1,3)- and (2,4)-semi-simple splicing systems were shown to be incomparable by Ceterchi et al.~\cite{Ceterchi2003}. In this section, we will consider the state complexity of languages generated by (2,4)-semi-simple splicing systems.

First, we define an NFA that recognizes the language of a given (2,4)-semi-simple splicing system. 
This construction is based on the construction of Head and Pixton~\cite{Head2006} for P\u{a}un splicing rules, which is based on the construction by Pixton~\cite{Pixton1996} for Pixton splicing rules. The original proof of regularity of finite splicing is due to Culik and Harju~\cite{Culik1991}. We follow the Head and Pixton construction and apply $\varepsilon$-transition removal on the resulting NFA to obtain an NFA for the semi-simple splicing system with the same number of states as the DFA for the initial language of the splicing system.

\begin{proposition} \label{prop:pixton-construction}
Let $H = (\Sigma,M,L)$ be a (2,4)-semi-simple splicing system with a regular initial language and let $L$ be recognized by a DFA with $n$ states. Then there exists an NFA $A_H'$ with $n$ states such that $L(A_H') = L(H)$.
\end{proposition}

\begin{proof} 
Let $H = (\Sigma,M,L)$ and let $A = (Q,\Sigma,\delta,q_0,F)$ be a DFA for $L$, with $|Q| = n$. Recall that markers $(a,b)$ correspond to a splicing rule $(\varepsilon, a; \varepsilon, b)$. For each marker $(a,b) \in M$, let $B_{(a,b)}$ be an automaton with initial state $i_{(a,b)}$ and final state $t_{(a,b)}$ which accepts the word $b$. The automaton $B_{(a,b)}$ is called a bridge for $(a,b)$ and is shown in Figure~\ref{fig:bridge}.

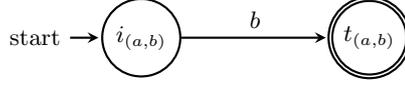
\begin{figure}
\begin{center}
\begin{tikzpicture}[shorten >=1pt,on grid,node distance=3cm,>=stealth,thick,auto]
    \node[state,initial] (0) {$i_{(a,b)}$};
    \node[state,accepting] (1) [right of=0]{$t_{(a,b)}$};
    \path[->]
        (0) edge  node {$b$} (1)
        ;
\end{tikzpicture}
\caption{The bridge $B_{(a,b)}$ for (2,4)-semi-simple splicing rules}
\label{fig:bridge}
\end{center}
\end{figure}

We will define the NFA $A_H = (Q',\Sigma,\delta',q_0,F)$, where the state set is  
\begin{equation*}
Q' = Q \cup \bigcup_{(a,b) \in M} \left\{i_{(a,b)}, t_{(a,b)} \right\}.
\end{equation*} 
Note that the initial and final states of $A_H$ stay unchanged despite the addition of states from the automata $B_{(a,b)}$.

We will now construct $\delta'$. First, we define $\delta_0'$ by
\begin{equation*}
\delta_0' = \delta \cup \bigcup_{(a,b) \in M} \left\{i_{(a,b)} \xrightarrow{b} t_{(a,b)}\right\}.
\end{equation*}
Then, we define $\delta_k'$ recursively for $k > 0$ by adding new transitions to $\delta_{k-1}'$ in the following way. For each marker $(a,b) \in M$,
\begin{enumerate}
    \item if $q \in Q'$ and $q \neq t_{(a',b')}$ for any $(a',b') \in M$, and
    \item $\delta_{k-1}'(q,a)$ is defined and useful,
\end{enumerate}
add a transition $q \xrightarrow{\varepsilon} i_{(a,b)}$ to $\delta_k'$, and,
\begin{enumerate}
    \item If $q \in Q'$ and $q \neq i_{(a',b')}$ for any $(a',b') \in M$, and
    \item $q \in \im (\delta_{k-1}')_b$,
\end{enumerate}
add a transition $t_{(a,b)} \xrightarrow{\varepsilon} q$ to $\delta_k$.

Since there are finitely many states, there can be only finitely many $\varepsilon$-transitions that can be added at each iteration and therefore there exists some $k$ for which $\delta_k' = \delta_{k+1}'$.  

We will now show that for (2,4)-semi-simple splicing, we have $k = 2$. Observe that $\delta_1'$ consists of all $\varepsilon$-transitions that either go from states of the original DFA to a bridge or transitions that go from a bridge to states of the original DFA. Then the only transitions that are in $\delta_2'$ which are not already present in $\delta_1'$ are $\varepsilon$-transitions of the form $i_{(a,b)} \xrightarrow{\varepsilon} i_{(a,b')}$ where $(a,b), (a,b') \in M$ and $t_{(a,b)} \xrightarrow{\varepsilon} t_{(a',b)}$ for $(a,b),(a',b) \in M$. From this, it is clear that no other $\varepsilon$-transitions can be added and therefore we have $\delta_2' = \delta_3'$. Therefore, by construction we have $L(A_H) = L(H)$.

To see this, informally, we can consider the path in $L(A_H)$ of a word $w \in L(H)$ and suppose that $w = uv$ is the result of a splicing action on the marker $(a,b)$, with $b$ being the first symbol of $v$. Such a path would trace $u$ from the initial state until it reaches a state $q$ with an outgoing transition on $a$. By construction, there is a $\varepsilon$-transition from $q$ to the state $i_{(a,b)}$. From $i_{(a,b)}$, following the transition on $b$ takes us to state $t_{(a,b)}$, from which there are $\varepsilon$-transitions to all states with an incoming transition on $b$. Since $v$ begins with $w$, the rest of the path follows the path corresponding to the rest of the word $v$ to an accepting state. Since $\varepsilon$-transitions are added for states that are on an accepting path (that is, those states that are useful), this process can be repeated several times before reaching an accepting state.

Finally, we can simplify this NFA by removing $\varepsilon$-transitions in the usual way to obtain an NFA $A_H' = (Q,\Sigma,\delta',q_0,F)$, where
\begin{equation*}
\delta'(q,b) = \begin{cases}
\im \delta_b & \text{if $(a,b) \in M$ and $\delta(q,a)$ is useful,} \\
\{\delta(q,b)\} & \text{otherwise.}
\end{cases}
\end{equation*}
Observe that by removing the $\varepsilon$-transitions, we also remove the states that were initially added earlier in the construction of $A_H$. Thus, the state set of $A_H'$ is exactly the state set of the DFA $A$ recognizing $L$.
\qed
\end{proof}

From this NFA construction, we can obtain a DFA via subset construction. This gives an upper bound of $2^n-1$ reachable states. This upper bound is the same for (1,3)-simple and (1,3)-semi-simple splicing systems and was shown to be tight~\cite{Kari2019a}. Since (1,3)-simple splicing systems and (2,4)-simple splicing systems are equivalent, we state without proof that the same result holds for (2,4)-simple splicing systems via the same lower bound witness. Therefore, this bound is reachable for (2,4)-semi-simple splicing systems via the same lower bound witness.

\begin{proposition}[\cite{Kari2019a}]
For $|\Sigma| \geq 3$ and $n \geq 3$, there exists a (2,4)-simple splicing system with a regular initial language $H = (\Sigma,M,L_n)$ with $|M| = 1$ where $L_n$ is a regular language with state complexity $n$ such that the minimal DFA for $L(H)$ requires at least $2^n - 1$ states.
\end{proposition}

It was also shown in~\cite{Kari2019a} that if the initial language is finite, this upper bound is not reachable for (1,3)-simple and (1,3)-semi-simple splicing systems. This result holds for all variants of semi-simple splicing systems and the proof is exactly the same as in~\cite{Kari2019a}. We state the result for semi-simple splicing systems and include the proof for completeness.
\begin{proposition}[\cite{Kari2019a}] \label{prop:finite-upper}
Let $H = (\Sigma, M, I)$ be a semi-simple splicing system with a finite initial language where $I$ is a finite language recognized by a DFA $A$ with $n$ states. Then a DFA recognizing $L(H)$ requires at most $2^{n-2} + 1$ states.
\end{proposition}

\begin{proof}
Let $A = (Q,\Sigma,\delta,q_0,F)$ and let $A_H$ be the DFA recognizing $L(H)$ obtained via the construction from Proposition~\ref{prop:pixton-construction}.  We will show that not all $2^n - 1$ non-empty subsets of $Q$ are reachable in $A_H$. First, since $I$ is a finite language, its DFA $A$ is acyclic. Therefore, $q_0$, the initial state of $A$, has no incoming transitions and thus the only reachable subset containing $q_0$ is $\{q_0\}$. Secondly, since $I$ is finite, $A$ must contain a sink state, which we will call $q_\emptyset$. Note that for any subset $P \subseteq Q$, we have that $P$ and $P \cup \{q_\emptyset\}$ are indistinguishable and can be merged together. This gives us a total of $2^{n-2} - 1 + 2$ states. 
\qed
\end{proof}

We will show that the bound of Proposition~\ref{prop:finite-upper} is reachable for (2,4)-semi-simple splicing systems.
\begin{lemma} \label{lem:finite-lower}
There exists a (2,4)-semi-simple splicing system with a finite initial language $H = (\Sigma,M,I_n)$ where $I_n$ is a finite language with state complexity $n$ such that a DFA recognizing $L(H)$ requires $2^{n-2} + 1$ states.
\end{lemma}

\begin{proof}
We take $\Sigma = \Sigma_n$ and $M = M_n$ and construct the DFA $A_n = (Q_n,\Sigma_n,\delta_n,0,F_n)$ recognizing $I_n$, where $Q_n = \{0,\dots,n-1\}$, $\Sigma_n = \{b\} \cup \Gamma_n$ where $\Gamma_n = \{a_S \mid S \subseteq \{2,\dots,n-2\} \}$, and $F_n = \{n-2\}$. Then we define $\delta_n$ by
\begin{itemize}
\item $\delta_n(i,a_S) = \min\{j \in S \mid i < j \leq n-2\}$ for $1 \leq i \leq n-3$ and $S \subseteq \{2, \dots, n-2\}$,
\item $\delta_n(0,a_S) = 1$ for all $a_S \in \Gamma_n$,
\item $\delta_n(i,b) = i+1$ for $0 \leq i \leq n-3$,
\item $\delta_n(i,a) = n-1$ for all $a \in \Sigma_n$ and $i \in \{n-2,n-1\}$.
\end{itemize} 
Then we consider the (2,4)-semi-simple splicing system $H = (\Sigma_n, M_n, A_n)$ with $M_n = \{b\} \times \Gamma_n$. Consider the NFA recognizing $L(H)$ obtained via the construction from Proposition~\ref{prop:pixton-construction} and let $A_n'$ be the DFA that results from applying the subset construction. 

Since $(b,a_S) \in M$ and $\delta(i,b) \neq n-1$ for all $i < n-2$, by the definition of $A_n$, we can reach any subset $S \cup \{1\}$ with $S \subseteq \{2,\dots,n-2\}$ from the initial state $\{0\}$ via the symbol $a_S$. We will show that from each of these states, we can reach a state $T = \{i_1, \dots, i_k\}$ where $2 \leq i_1 < \cdots < i_k \leq n-2$. First, if $i_1 = 2$, then we let $T' = \{i_2 - 1, \dots, i_k - 1\}$ and the subset $T$ is reachable from the initial state via the word $a_{T'} b$. Otherwise, if $i_1 > 2$, then the subset $T$ is reachable from the initial state via the word $a_{T' \cup \{i_1 - 1\}} b$.

To show that each of these states is pairwise distinguishable, first we note that $\{0\}$ is distinguishable from every other state. Now suppose that we have two subsets $S, S' \subseteq \{1,\dots,n-2\}$ such that $S \neq S'$. Without loss of generality, there is a state $t \in S$ such that $t \not \in S'$. Then these two states can be distinguished by the word $b^{n-2-t}$. This gives us $2^{n-2} - 1$ states.

For the last two states, we see that $\{0\}$ is reached on the word $\varepsilon$ and it is clearly distinguishable from every other state. The sink state $\{n-1\}$ is reachable via the word $b^{n-1}$ and is distinguishable since it is the sole sink state of the machine. Thus, in total $A_n'$ requires $2^{n-2} + 1$ states.
\qed
\end{proof}

Here, our lower bound example requires an alphabet that grows exponentially with the number of states. We will show in the following that this is necessary. This is in contrast to the lower bound witness for (1,3)-semi-simple systems from~\cite{Kari2019a}, which requires only three letters. We also note that the initial language used for this witness is the same as that for (1,3)-simple splicing systems from~\cite{Kari2019a}. From this, we observe that the choice of the visible sites for the splicing rules (i.e. (1,3) vs. (2,4)) makes a difference in the state complexity. We will see other examples of this later as we consider semi-simple splicing systems with other rule variants.
\begin{lemma}
Let $H = (\Sigma,M,I)$ be a (2,4)-semi-simple splicing system with a finite initial language where $I$ is a finite language with state complexity $n$. If the DFA recognizing $L(H)$ requires $2^{n-2} + 1$ states, then $|\Sigma| \geq 2^{n-3}$.
\end{lemma}

\begin{proof}
Let $A = (Q,\Sigma,\delta,q_0,F)$ be a DFA with $n$ states recognizing $I$. Since $I$ is a finite language, there exists at least one state of $A$, say $q_1$ that is reachable only from the initial state $q_0$. Let $A'$ be the DFA obtained via applying the subset construction to the NFA for $L(H)$ obtained via the construction of Proposition~\ref{prop:pixton-construction}. Then any subset $P \subseteq Q$ with $q_1 \in P$ and $|P| \geq 2$ can only be reached in $A'$ via a transition on a symbol $b$ with $(a,b) \in M$. However, there can be up to $2^{n-3}$ subsets of $Q$ that contain $q_1$. Therefore, $\Sigma$ must contain at least $2^{n-3}$ symbols.
\qed
\end{proof}

Together, Proposition~\ref{prop:finite-upper} and Lemma~\ref{lem:finite-lower} give the following result.
\begin{theorem}
Let $H = (\Sigma,M,I)$ be a (2,4)-semi-simple splicing system with a finite initial language, where $I$ is a finite language with state complexity $n$ and $M \subseteq \Sigma \times \Sigma$. Then the state complexity of $L(H)$ is at most $2^{n-2} + 1$ and this bound can be reached in the worst case.
\end{theorem}

\section{State complexity of (2,3)-semi-simple splicing systems} \label{sec:23}
We will now consider the state complexity of (2,3)-semi-simple splicing systems. Recall that a (2,3)-semi-simple splicing system is a splicing system with rules of the form $(\varepsilon,a;b,\varepsilon)$ for $a, b \in \Sigma$. We can follow the same construction from Proposition~\ref{prop:pixton-construction} with slight modifications to account for $(2,3)$-semi-simple splicing rules to obtain an NFA for a language generated by a (2,3)-semi-simple splicing system with the same number of states as the DFA for the initial language of the splicing system.

\begin{proposition} \label{prop:23-construction}
Let $H = (\Sigma,M,L)$ be a (2,3)-semi-simple splicing system with a regular initial language and let $L$ be recognized by a DFA with $n$ states. Then there exists an NFA $A_H'$ with $n$ states such that $L(A_H') = L(H)$.
\end{proposition}

\begin{proof}
Let $H = (\Sigma,M,L)$ be a (2,3)-semi-simple splicing system with a regular initial language, where $M \subseteq \Sigma \times \Sigma$ and $L \subseteq \Sigma^*$, and let $A = (Q,\Sigma,\delta,q_0,F)$ be a DFA that recognizes $L$. We will define the NFA $A_H = (Q',\Sigma,\delta',q_0,F)$ by following the construction of Proposition~\ref{prop:pixton-construction} with a modification to the definition of the bridges $B_{(a,b)}$.

For (2,3)-semi-simple splicing, for each marker $(a,b) \in M$, we define the bridge $B_{(a,b)}$ as an automaton with initial state $i_{(a,b)}$, final state $t_{(a,b)}$, and a transition $i_{(a,b)} \xrightarrow{\varepsilon} t_{(a,b)}$. The bridge $B_{(a,b)}$ for (2,3)-semi-simple splicing rules is shown in Figure~\ref{fig:bridge23}.

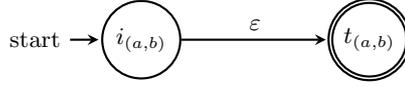
\begin{figure}
\begin{center}
\begin{tikzpicture}[shorten >=1pt,on grid,node distance=3cm,>=stealth,thick,auto]
    \node[state,initial] (0) {$i_{(a,b)}$};
    \node[state,accepting] (1) [right of=0]{$t_{(a,b)}$};
    \path[->]
        (0) edge  node {$\varepsilon$} (1)
        ;
\end{tikzpicture}
\caption{The bridge $B_{(a,b)}$ for (2,3)-semi-simple splicing rules}
\label{fig:bridge23}
\end{center}
\end{figure}

We define the transition function $\delta'$ in the same way as in the construction of Proposition~\ref{prop:pixton-construction} and note that for (2,3)-semi-simple splicing systems, only one iteration is necessary. That is, $\delta' = \delta_1' = \delta_2'$, where for $a,b \in \Sigma$ and $q \in Q'$,
\begin{align*}
\delta_0' &= \delta \cup \bigcup_{(a,b) \in M} \left\{ i_{(a,b)} \xrightarrow{\varepsilon} t_{(a,b)} \right\}\text{, and,} \\
\delta_1' &= \delta_0' \cup \bigcup_{(a,b) \in M} (\{q \xrightarrow{\varepsilon} i_{(a,b)} \mid \text{$\delta_0'(q,a)$ is useful}\} \cup \{t_{(a,b)} \xrightarrow{\varepsilon} \im (\delta_0')_b\}).
\end{align*}

Since the only transitions that are not $\varepsilon$-transitions are between states of the original NFA, it is clear that after the first iteration of additional $\varepsilon$-transition, no further $\varepsilon$-transitions may be added. Therefore, by construction, we have $L(A_H) = L(A_H')$.
Furthermore, it is clear that removing $\varepsilon$-transitions from $A_H$ will result in an NFA $A_H'$ that has a state set $Q$, the state set of $A$.
\end{proof}

From Proposition~\ref{prop:23-construction}, we get an upper bound of $2^n - 1$ reachable states via the subset construction. However, we will show that this bound cannot be reached.


\begin{proposition} \label{prop:23-upper}
Let $H = (\Sigma,M,L)$ be a (2,3)-semi-simple splicing system with a regular initial language, where $M \subseteq \Sigma \times \Sigma$ and $L \subseteq \Sigma^*$ is recognized by a DFA with $n$ states. Then there exists a DFA $A_H$ such that $L(A_H) = L(H)$ and $A_H$ has at most $2^{n-1}$ states.
\end{proposition}

\begin{proof}
Let $A = (Q,\Sigma,\delta,q_0,F)$ be the DFA for $L$ and let $B_H = (Q,\Sigma,\delta',q_0,F)$ be the NFA obtained via the construction given the (2,3)-semi-simple splicing system $H$. Let $A_H$ be the DFA obtained by applying the subset construction to $B_H$. Note that the states of $A_H$ are subsets of states of $B_H$. 

Recall that for each symbol $a \in \Sigma$ for which there is a pair $(a,b) \in M$, if the machine $B_H$ enters a state $q \in Q$ with an outgoing transition on $a$, the machine $B_H$ also simultaneously enters any state with an incoming transition on~$b$. Consider $a \in \Sigma$ with $(a,b) \in M$ and $\delta(q,a) = q'$ for some $q' \in Q$. Assuming that $(a,b)$ is non-trivial and $\im \delta_b$ contains useful states, for any set $P \subseteq Q$, we must have $\im \delta_b \subseteq P$ if $q \in P$. This implies that not all $2^n - 1$ non-empty subsets of $Q$ are reachable, since that would mean the singleton set $\{q\}$ is unreachable.  

Thus, to maximize the number of sets that can be reached, the number of states with incoming transitions on any symbol $b$ with $(a,b) \in M$ must be minimized. Therefore, for $(a,b) \in M$, there can be only one useful state with incoming transitions on $b$. Let us call this state $q_b \in Q$. 

We claim that to maximize the number of states, $A$ must contain no useless states and therefore $A$ contains no sink state. First, suppose otherwise and that $A$ contains a sink state $q_\emptyset$. To maximize the number of states, we minimize the number of states of $A$ with outgoing transitions, so there is only one state of $A$, say $q'$, with an outgoing transition on $a$. We observe that $q' \neq q_b$, since otherwise, $|\im \delta_b| = 1$ and the only reachable subset that contains $q_b$ is the singleton set $\{q_b\}$.

Now, recall that for all subsets $P \subseteq Q \setminus \{q_\emptyset\}$, the two sets $P$ and $P \cup \{q_\emptyset\}$ are indistinguishable. Then there are at most $2^{n-2}$ subsets containing $q_b$ and at most $2^{n-3}-1$ nonempty subsets of $Q \setminus \{q_b,q',q_\emptyset\}$. Together with the sink state, this gives a total of at most $2^{n-2} + 2^{n-3}$ states in~$A_H$.

Now, we consider when $A$ contains no sink state. In this case, since $A$ must be a complete DFA, in order to satisfy the condition that $|\im \delta_b|$ is minimal, we must have $\delta(q,a) = q_b$ for all $q \in Q$. But this means that for any state $q \in Q$ and subset $P \subseteq Q$, if $q \in P$, then $q_b \in P$. Therefore, every reachable subset of $Q$ must contain $q_b$. This gives an upper bound of $2^{n-1}$ states in~$A_H$.

Since $2^{n-1} > 2^{n-2} + 2^{n-3}$ for $n \geq 3$, the DFA $A_H$ can have at most $2^{n-1}$ states in the worst case.
\qed
\end{proof}

This bound is reachable when the initial language is a regular language, even when restricted to simple splicing rules defined over an alphabet of size 3.
\begin{lemma} \label{lem:23-lower}
There exists a (2,3)-simple splicing system with a regular initial language $H = (\Sigma,M,L_n)$ with $|\Sigma| = 3$, $|M| = 1$, and $L_n$ is a regular language with state complexity $n$ such that the minimal DFA for $L(H)$ requires at least $2^{n-1}$ states.
\end{lemma}

\begin{proof}
Let $L_n$ be the language recognized by the DFA $A_n = (Q_n,\Sigma,\delta_n,0,F_n)$, where $Q_n = \{0,1,\dots,n-1\}$, $F_n = \{n-1\}$, and the transition function $\delta_n$ is defined by
\begin{itemize}
\item $\delta_n(i,a) = i+1 \bmod n$ for $0 \leq i \leq n-1$,
\item $\delta_n(0,b) = b$, $\delta_n(1,b) = 0$, $\delta_n(i,b) = i$ for $2 \leq i \leq n-1$,
\item $\delta_n(i,c) = 0$ for $0 \leq i \leq n-1$.
\end{itemize}
The DFA $A_n$ is shown in Figure~\ref{fig:23-witness}.

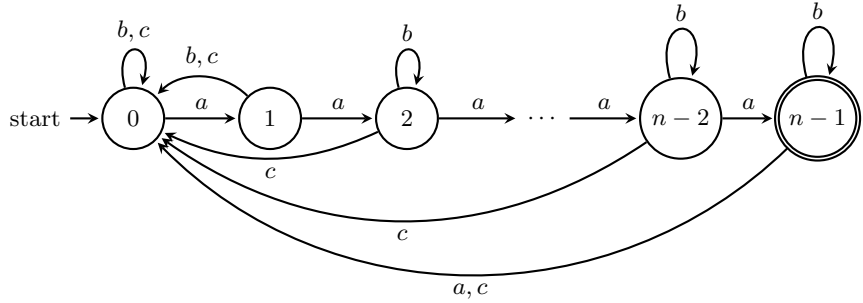
\begin{figure}
\begin{center}
\begin{tikzpicture}[shorten >=1pt,on grid,node distance=1.8cm,>=stealth,thick,auto]
    \node[state,initial] (0) {$0$};
    \node[state] (1) [right of=0]{$1$};
    \node[state] (2) [right of=1]{$2$};
    \node (3) [right of=2]{$\cdots$};
    \node[state] (4) [right of=3]{${{n-2}}$};
    \node[state,accepting] (5) [right of=4]{${{n-1}}$};
    \path[->]
        (0) edge [loop above] node {$b,c$} (0)
        (2) edge [loop above] node {$b$} (2)
        (4) edge [loop above] node {$b$} (4)
        (5) edge [loop above] node {$b$} (5)
        (0) edge  node {$a$} (1)
        (1) edge [bend right=45,above] node {$b,c$} (0)
        (1) edge  node {$a$} (2)
        (2) edge  node {$a$} (3)
        (3) edge  node {$a$} (4)
        (4) edge  node {$a$} (5)
        (2) edge [bend left=25] node {$c$} (0)
        (4) edge [bend left=35] node {$c$} (0)
        (5) edge [bend left=45] node {$a,c$} (0)
        ;
\end{tikzpicture}
\caption{The DFA $A_n$ of Lemma~\ref{lem:23-lower}}
\label{fig:23-witness}
\end{center}
\end{figure}
Consider the (2,3)-simple splicing system $H = (\Sigma,\{(c,c)\},L_n)$ and consider the DFA obtained via the construction of Proposition~\ref{prop:23-upper}. Then the states of $A_n'$ are subsets of $Q_n$. Observe that by definition of $A_n$ and $H$, every reachable subset of $A_n'$ must contain $0$. We will show that all states $P \subseteq Q$ with $0 \in P$ are reachable and pairwise distinguishable.

First, the initial state $\{0\}$ is clearly reachable. We will show that for $S \subseteq \{1,\dots,n-1\}$, all states $\{0\} \cup S$ are reachable, by induction on the size of $S$. First, for $|S| = 1$, we have
\begin{equation*}
\{0\} \xrightarrow{a(ab)^{i-1}} \{0,i\}
\end{equation*}
for $1 \leq i \leq n-1$. Thus, all sets $\{0\} \cup S$ with $|S| = 1$ are reachable. Now, for $k \geq 2$, suppose that all sets $\{0\} \cup S$ with $|S| = k$ are reachable. We will show that sets of size $k+1$ are reachable. Let $T = \{0, i_1,\dots,i_{k+1} \}$ with $0 < i_1 < \cdots < i_{k+1} \leq n-1$. Then,
\begin{equation*}
\{0, i_2 - i_1, \dots, i_{k+1} - i_1 \} \xrightarrow{a(ab)^{i_1-1}} \{0, i_1, \dots, i_{k+1} \}.
\end{equation*}
Thus, all sets $\{0\} \cup S$ with $|S| = k+1$ are reachable and therefore all sets $\{0\} \cup S$ with $S \subseteq \{1,\dots,n-1\}$ are reachable.

To see that each of these states is pairwise distinguishable, consider two subsets $P,P' \subseteq Q$ with $P \neq P'$. Then there is an element $t \in Q$ such that $t \in P$ and $t \not \in P'$ and these two states can be distinguished by the word $a^{n-1-t}$. 

Thus, we have shown that $A_n'$ contains at least $2^{n-1}$ reachable and distinguishable states.
\qed
\end{proof}

Together, Proposition~\ref{prop:23-upper} and Lemma~\ref{lem:23-lower} give the following result.
\begin{theorem}
Let $H = (\Sigma,M,L)$ be a (2,3)-semi-simple splicing system with a regular initial language, where $L \subseteq \Sigma^*$ is a regular language with state complexity $n$ and $M \subseteq \Sigma \times \Sigma$. Then the state complexity of $L(H)$ is at most $2^{n-1}$ and this bound can be reached in the worst case.
\end{theorem}

Recall that in the proof of Proposition~\ref{prop:23-upper}, the bound depended on whether or not the DFA for the initial language contained a sink state. Since a DFA recognizing a finite language must have a sink state, the upper bound stated in the proposition is clearly not reachable when the initial language is finite. We will show that, in fact, even the upper bound in the case where the DFA contains a sink state which was stated in the proof is not reachable when the initial language is finite.
\begin{proposition} \label{prop:23-semi-finite-upper}
Let $H = (\Sigma,M,I)$ be a (2,3)-semi-simple splicing system where $I$ is a finite language recognized by a DFA $A$ with $n$ states. Then a DFA recognizing $L(H)$ requires at most $2^{n-3} + 2$ states.
\end{proposition}

\begin{proof} 
Let $A = (Q,\Sigma,\delta,q_0,F)$ be the DFA for $I$ and let $A_H$ be the DFA obtained via the construction of Proposition~\ref{prop:23-upper}, given the (2,3)-semi-simple splicing system $H$. We will consider the number of reachable and pairwise distinguishable states of $A_H$. 

Recall from the proof of Proposition~\ref{prop:23-upper} that to maximize the number of sets that can be reached in $A_H$, the number of states with incoming transitions on any symbol $b$ with $(a,b) \in M$ must be minimized. Then for $(a,b) \in M$, there can be only one useful state with incoming transitions on $b$. Let us call this state $q_b \in Q$.

Since $I$ is a finite language, we know that $q_0$, the initial state of $A$, is contained in exactly one reachable state in $A_H$. Similarly $A$ must contain a sink state $q_\emptyset$ and for all subsets $P \subseteq Q$, we have that $P$ and $P \cup \{q_\emptyset\}$ are indistinguishable. Finally, we observe that there must exist at least one state $q_1 \in Q$ that is directly reachable from $q_0$ and is not reachable by any word of length greater than 1. Therefore, in order to maximize the number of reachable subsets, we must have that $q_1 = q_b$. 

Let $Q^a$ denote the set of states for which there is an outgoing transition on the symbol $a$. That is, if $q \in Q^a$, we have $\delta(q,a) \leq n-2$. Let $k_a = |Q^a|$. It is clear that $k_a \geq 1$. Now, consider a reachable subset $P \subseteq Q \setminus \{q_0,q_\emptyset\}$. We claim that if $|P| \geq 2$ and $q_b \in P$, then we must have $q \in P$ for some $q \in Q^a$. 

Suppose otherwise and that $Q^a \cap P = \emptyset$. Recall that $q_b = q_1$ and the only incoming transitions to $q_1$ are from the initial state $q_0$. Then this means that $P = \{q_1\}$ and $|P| = 1$, a contradiction. Therefore, we have $Q^a \cap P \neq \emptyset$ whenever $q_b \in P$ with $|P| \geq 2$.

Now, we can count the number of reachable subsets of $Q \setminus \{q_0,q_\emptyset\}$. There are $2^{n-3-k_a} (2^{k_a} - 1)$ non-empty subsets of size greater than 1 which contain $q_b$ and there are $2^{n-3-k_a} - 1$ non-empty subsets which do not contain $q_b$. Together with the initial and sink states and the set $\{q_b\}$, we have
\begin{equation*}
2^{n-3-k_a} (2^{k_a} - 1) + 2^{n-3-k_a} - 1 + 3.
\end{equation*}
Thus, the DFA $A_H$ has at most $2^{n-3} + 2$ reachable states.
\qed
\end{proof}

We will show that there exists a (2,3)-semi-simple splicing system with initial finite languages defined over a fixed alphabet that can reach the upper bound.
\begin{lemma} \label{lem:23-semi-finite-lower}
There exists a (2,3)-semi-simple splicing system with a finite initial language $H = (\Sigma,M,I_n)$ with $|M| = 1$ where $I_n$ is a finite language with state complexity $n$ such that the minimal DFA for $L(H)$ requires at least $2^{n-3} + 2$ states.
\end{lemma}

\begin{proof}
Let $A_n = (Q_n, \Sigma, \delta_n, 0, \{n-2\})$, where $Q_n = \{0, \dots, n-1\}$, $\Sigma = \{a,b,c\}$, and the transition function is defined
\begin{itemize}
\item $\delta(i,a) = i+1$ for $1 \leq i \leq n-2$,
\item $\delta(0,b) = \delta(1,b) = n-1$, $\delta(i,b) = i+1$ for $2 \leq i \leq n-2$,
\item $\delta(0,c) = 1$, $\delta(i,c) = n-1$ for $1 \leq i \leq n-2$,
\item $\delta(q,\sigma) = n-1$ for all other transitions not already defined.
\end{itemize}
The DFA $A_n$ is shown in Figure~\ref{fig:23-semi-finite}.

\begin{figure}
\begin{center}
\begin{tikzpicture}[shorten >=1pt,on grid,node distance=1.8cm,>=stealth,thick,auto]
    \node[state,initial] (0) {$0$};
    \node[state] (1) [right of=0]{$1$};
    \node[state] (2) [right of=1]{$2$};
    \node (3) [right of=2]{$\cdots$};
    \node[state] (4) [right of=3]{${{n-3}}$};
    \node[state,accepting] (5) [right of=4]{${{n-2}}$};
    \path[->]
        (0) edge  node {$c$} (1)
        (1) edge  node {$a$} (2)
        (2) edge  node {$a,b$} (3)
        (3) edge  node {$a,b$} (4)
        (4) edge  node {$a,b$} (5)
        ;
\end{tikzpicture}
\caption{The DFA $A_n$ of Lemma~\ref{lem:23-semi-finite-lower}. Transitions not shown are to the sink state $n-1$, which is not shown.}
\label{fig:23-semi-finite}
\end{center}
\end{figure}
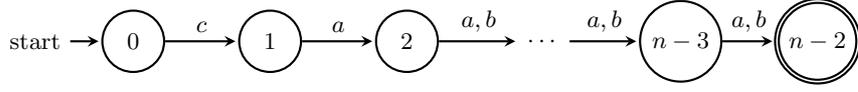

Let $H = (\Sigma,\{(a,c)\},I_n)$ be a (2,3)-semi-simple splicing system with a finite initial language. We apply the construction from Proposition~\ref{prop:23-upper} to obtain a DFA~$A_n'$. We will show that $A_n'$ has $2^{n-3} + 2$  reachable states. First, we observe that the initial state of $A_n'$ is $\{0\}$. We will consider the reachable subsets of $\{1,\dots,n-2\}$.

Observe that since every state $i$ with $2 \leq i \leq n-3$ has an outgoing transition on $a$, $1 \in S$ for all reachable subsets $S \subseteq \{1,\dots,n-2\}$, since $(a,c) \in M$ and $1$ is the sole state of $A_n$ with an incoming transition on $c$. For each set $T \subseteq \{2,\dots,n-2\}$, let $t = \max T$ and we define words $w_T = a_t a_{t-1} \cdots a_2$ of length $t-1$ by
\begin{equation*}
a_j = \begin{cases}
a & \text{if $j \in T$},\\
b & \text{otherwise},
\end{cases}
\end{equation*}
for $2 \leq j \leq n-2$. Then $\{0\} \xrightarrow{c} \{1\} \xrightarrow{w_T} \{1\} \cup T$ and all $2^{n-3}$ subsets $\{1\} \cup T$ with $T \subseteq \{2,\dots,n-2\}$ are reachable.

To see that every reachable state is pairwise distinguishable, consider two subsets $P,P' \subseteq Q$ with $P \neq P'$. Then there is an element $t \in Q$ such that $t \in P$ and $t \not \in P'$. These two subsets can then be distinguished via the word $a^{n-2-t}$.

Thus, we have shown that $A_n'$ has at most $2^{n-3} + 2$ reachable and pairwise distinguishable states.
\qed
\end{proof}

Proposition~\ref{prop:23-semi-finite-upper} and Lemma~\ref{lem:23-semi-finite-lower} give the following theorem.
\begin{theorem} \label{thm:23-semi-finite}
Let $H = (\Sigma,M,I)$ be a (2,3)-semi-simple splicing system with a finite initial language, where $I$ is a finite language with state complexity $n$ and $M \subseteq \Sigma \times \Sigma$. Then the state complexity of $L(H)$ is at most $2^{n-3} + 2$ and this bound can be reached in the worst case.
\end{theorem}

Unlike the situation with regular initial languages, when we restrict (2,3)-semi-simple splicing systems with initial finite languages further to allow only (2,3)-simple splicing rules, the bound of Theorem~\ref{thm:23-semi-finite} is not reachable.
\begin{proposition} \label{prop:23-finite-simple-upper}
Let $H = (\Sigma,M,I)$ be a (2,3)-simple splicing system where $I$ is a finite language recognized by a DFA $A$ with $n$ states. Then a DFA recognizing $L(H)$ requires at most $2^{n-4} + 2^{n-5} + 2$ states.
\end{proposition}

\begin{proof}
Let $A = (Q,\Sigma,\delta,q_0,F)$ be the DFA for $I$ and let $B_H = (Q,\Sigma,\delta,q_0,F)$ be the NFA obtained via the construction from Proposition~\ref{prop:23-construction} given the (2,3)-simple splicing system $H$. Let $A_H$ be the DFA obtained by applying the subset construction to $B_H$ and the states of $A_H$ are subsets of states of $B_H$. 

It is well known that DFAs recognizing finite languages are acyclic and that their states can be ordered. For an integer $i \geq 0$, let $Q_i$ be the set of states of $A$ that are reachable by a word of length at most $i$. For two states $p,q \in S$, we write $p < q$ if $p \in Q_i$ and $q \in Q_j$ with $i < j$.

For each symbol $a \in \Sigma$ with $(a,a) \in M$, let $Q^a = \{q \in Q \mid \delta(q,a) \neq q_\emptyset\}$, the set of states with outgoing transitions to non-sink states on the symbol $a$. Observe that for a subset $P \subseteq Q$, if $P$ is a reachable subset in $A_H$ and $P \cap Q^a \neq \emptyset$, then $\im \delta_a \subseteq P$. From these conditions, we have the following states.
\begin{enumerate}
\item The initial state $q_0' = \{q_0\}$ and sink state $q_\emptyset$.
\item States $P \subseteq Q \setminus \left(\{q_0,q_\emptyset\} \cup Q_1 \cup \bigcup_{(a,a) \in M} Q^a \right)$
\item For each $(a,a) \in M$, $\im \delta_a \cup P \cup P'$,  where $P \subseteq Q \setminus \left(\{q_0,q_\emptyset\} \cup \bigcup_{a \in M} Q^a \right)$ and $P' \subseteq Q^a$.
\end{enumerate}
To maximize the number of subsets of $Q$ that can be reached, we must minimize the number of states with incoming and outgoing transitions on markers (symbols in $M$) and assume that for all $(a,a), (b,b) \in M$ with $a \neq b$, the sets $Q^a$ and $Q^b$ are disjoint.

From this, it is clear that we must have $q_0 \in Q^a$ and at least one state $q_1 \in Q_1$ with $q_1 \in \im \delta_a$ for some $(a,a) \in M$. Suppose otherwise. Since states in $Q_1$ are reachable only from the initial state $q_0$, any reachable subset $P \subseteq Q$ with $|P| \geq 2$ can not contain a state from $Q_1$. In fact, for any state $q_1 \in Q_1$ and subset $P \subseteq Q$ of size 2 or greater, we have $q_1 \in P$ only if $\im \delta_a \subseteq P$ for some marker $a \in M$. Thus, there must exist a transition $\delta(q_0,a) = q_1$ for $a \in M$. 

However, this is insufficient. Since the initial state $q_0$ is only reachable on $\varepsilon$, in order to reach a subset $P \subseteq Q \setminus \{q_0\}$, there must exist at least one other state $q \neq q_0$ in $Q^a$. Furthermore, this state $q$ must have a transition on $a$ to some state $q' \not \in Q_1$. Thus, there are at least two states in $Q^a$ and there are at least two states in $\im \delta_a$.

For $(a,a) \in M$, let
\begin{equation*}
t_a = \begin{cases}
|Q^a| & \text{if $q_0 \not \in Q^a$,} \\
|Q^a| - 1 & \text{if $q_0 \in Q^a$,}
\end{cases}
\end{equation*}
and let 
\begin{equation*}
t = \begin{cases}
\left|\bigcup_{(a,a) \in M} Q^a \right| - 1 & \text{if $q_0 \in Q^a$ for some $(a,a) \in M$,} \\
\left|\bigcup_{(a,a) \in M} Q^a \right| & \text{otherwise.} \\
\end{cases}
\end{equation*}
Then (1) gives 2 states, (2) gives $2^{n-2-|Q_1|-t}$ states, and (3) gives up to 
\begin{equation*}
\sum_{a \in M} (2^{t_a}-1)(2^{n-2-|\im \delta_a| - t_a}).
\end{equation*}
Thus, to maximize the number of subsets of $Q$ that can be reached, both sets $Q^a$ and $\im \delta_a$ must be minimized. We have already shown above that at minimum, $|Q^a| = 2$ and $|\im \delta_a| = 2$. Then $t_a = 1$ and this gives up to $2 + 2^{n-4} + 2^{n-5}$ states.
\qed
\end{proof}

We will show that this bound is reachable by a family of witnesses defined over a fixed alphabet.
\begin{lemma} \label{lem:23-finite-simple-lower}
There exists a (2,3)-simple splicing system $H = (\Sigma,M,I_n)$ with $|M| = 1$ where $I_n$ is a finite language with state complexity $n$ such that the minimal DFA for $L(H)$ requires at least $2^{n-4} + 2^{n-5} + 2$ states.
\end{lemma}

\begin{proof}
Let $A_n = (Q_n, \Sigma, \delta_n, 0, \{n-2\})$, where $Q_n = \{0, \dots, n-1\}$, $\Sigma = \{a,b,c,d,e,f,g\}$, and the transition function is defined
\begin{align*}
&\delta_n(i,a) = i+1 \text{ for $0 \leq i \leq n-2,$} \\
&\delta_n(i,b) = i+1 \text{ for $i = 0, 1, 2, 4, \dots, n-2,$} & \delta_n(2,b) &= n-1, \\
&\delta_n(i,c) = i+1 \text{ for $i = 0, 1, 2,$} & \delta_n(i,c) &= n-1 \text{ for $3 \leq i \leq n-2,$} \\
&\delta_n(i,d) = n-1 \text{ for $i = 0, 1,$} & \delta_n(i,d) &= i+1 \text{ for $2 \leq i \leq n-2,$} \\
&\delta_n(i,e) = n-1 \text{ for $i = 0, 1, 3,$} & \delta_n(i,e) &= i+1 \text{ for $i = 2, 4, \dots, n-2,$} \\
&\delta_n(i,f) = n-1 \text{ for $i = 0, 1, 2,$} & \delta_n(i,f) &= i+1 \text{ for $3 \leq i \leq n-2,$} \\
&\delta_n(i,g) = n-1 \text{ for $i = 0, 1, 2, 3,$} & \delta_n(i,g) &= i+1 \text{ for $4 \leq i \leq n-2,$} \\
&\delta_n(n-1,\sigma) = n-1 \text{ for all $\sigma \in \Sigma$.} 
\end{align*}
The DFA $A_n$ is shown in Figure~\ref{fig:23-finite}.

\begin{figure}
\begin{center}
\begin{tikzpicture}[shorten >=1pt,on grid,node distance=1.8cm,>=stealth,thick,auto]
    \node[state,initial] (0) {$0$};
    \node[state] (1) [below left of=0]{$1$};
    \node[state] (2) [below right of=1]{$2$};
    \node[state] (3) [above right of=2]{$3$};
    \node[state] (4) [right of=3]{$4$};
    \node (a) [right of=4]{$\dots$};
    \node[state,accepting] (5) [right of=a]{${{n-2}}$};
    \path[->]
        (0) edge  node {$a,b,c$} (1)
        (1) edge [below left] node {$a,b,c$} (2)
        (2) edge [below right] node {$a,b,c,d,e$} (3)
        (3) edge [align=center] node {$a,d,f$} (4)
        (4) edge [align=center] node {$a,b,d$\\$e,f,g$} (a)
        (a) edge [align=center] node {$a,b,d$\\$e,f,g$} (5)
        ;
\end{tikzpicture}
\caption{The DFA $A_n$ of Lemma~\ref{lem:23-finite-simple-lower}. Transitions not shown are to the sink state $n-1$, which is not shown.}
\label{fig:23-finite}
\end{center}
\end{figure}
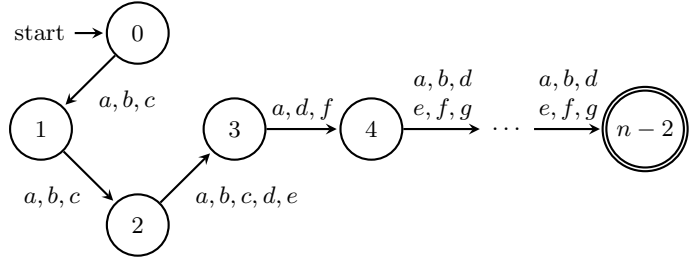

Let $H = (\Sigma,\{(c,c)\},I_n)$ be a (2,3)-simple splicing system with a finite initial language. We apply the construction from Proposition~\ref{prop:23-upper} to obtain a DFA $A_n'$. We will show that $A_n'$ has $2^{n-3} + 2^{n-4} + 2$ reachable states. First, we observe that the initial state of $A_n'$ is $\{0,1,2,3\}$, since $c \in M$. We will consider the reachable subsets of $\{1, \dots, n-2\}$. 

We will consider two cases. First, we will consider subsets $S \subseteq \{1, \dots, n-2\}$. By our construction, if $1 \in S$, then $2, 3 \in S$ and similarly, if $2 \in S$, then $1, 3 \in S$. Let $T = \{i_1, \dots, i_k\}$, where $4 \leq i_1 < \cdots < i_k \leq n-2$. We define words $w_T = a_4 \cdots a_{n-2}$ of length $n-5$ by
\begin{equation*}
a_j = \begin{cases}
a & \text{if $n-2 - j + 4 \in T$}, \\
b & \text{otherwise}.
\end{cases}
\end{equation*}
for $4 \leq j \leq n-2$. Then we have $\{0,1,2,3\} \xrightarrow{w_T} \{1,2,3\} \cup T$. 

Next, we will consider subsets $S \subseteq \{3, \dots, n-2\}$. Let $S = \{i_1, \dots, i_k\}$ with $3 \leq i_1 < \cdots < i_k \leq n-2$. There are four cases to consider.
\begin{itemize}
\item If $i_1 = 3$ and $i_2 = 4$, then $S$ is reachable from the state $\{1,2,3,i_3-1,\dots,i_k-1\}$ on the word $d$.
\item If $i_1 = 3$ and $i_2 > 4$, then $S$ is reachable from the state $\{1,2,3,i_2-1, \dots, i_k-1\}$ on the word $e$.
\item If $i_1 = 4$, then $S$ is reachable from the state $\{1,2,3,i_2-1, \dots, i_k-1\}$ on the word $f$.
\item If $i_1 > 4$, then $S$ is reachable from the state $\{1,2,3,i_1-1, \dots, i_k-1\}$ on the word $g$.
\end{itemize}

This gives a total of $2^{n-4} + 2^{n-5} + 2$ reachable states. To see that each of these states is pairwise distinguishable, consider two subsets $P,P' \subseteq Q$ with $P \neq P'$. Then there is an element $t \in Q$ such that $t \in P$ and $t \not \in P'$ and these two states can be distinguished by the word $a^{n-2-t}$.

Thus, we have shown that there are $2^{n-4} + 2^{n-5} + 2$ reachable and pairwise distinguishable states.
\qed
\end{proof}

Proposition~\ref{prop:23-finite-simple-upper} and Lemma~\ref{lem:23-finite-simple-lower} give the following theorem.
\begin{theorem} \label{thm:23-finite-simple}
Let $H = (\Sigma,M,I)$ be a (2,3)-simple splicing system with a finite initial language, where $I \subseteq \Sigma^*$ is a finite language with state complexity $n$ and $M \subseteq \Sigma^* \times \Sigma^*$. Then the state complexity of $L(H)$ is at most $2^{n-4} + 2^{n-5} + 2$ and this bound can be reached in the worst case.
\end{theorem}

\section{State Complexity of (1,4)-semi-simple Splicing Systems} \label{sec:14}
In this section, we consider the state complexity of (1,4)-semi-simple splicing systems. Recall that a (1,4)-semi-simple splicing system is a splicing system with rules of the form $(a,\varepsilon;\varepsilon,b)$ for $a, b \in \Sigma$. As with (2,3)-semi-simple splicing systems, we can easily modify the construction of Proposition~\ref{prop:pixton-construction} to obtain an NFA for (1,4)-semi-simple splicing systems.

\begin{proposition} \label{prop:14-construction}
Let $H = (\Sigma,M,L)$ be a (1,4)-semi-simple splicing system with a regular initial language, $M = M_1 \times M_2$ with $M_1,M_2 \subseteq \Sigma$ and let $L$ be recognized by a DFA with $n$ states. Then there exists an NFA $A_H'$ with $n+m$ states such that $L(A_H') = L(H)$, where $m = |M_1|$.
\end{proposition}

\begin{proof}
Let $H = (\Sigma,M,L)$ be a (1,4)-semi-simple splicing system with a regular initial language and let $A = (Q,\Sigma,\delta,q_0,F)$ be a DFA for $L \subseteq \Sigma^*$.  We will define the NFA $A_H = (Q',\Sigma,\delta',q_0,F)$ by following the construction of Proposition~\ref{prop:pixton-construction} with a modification to the definition of the bridges $B_{(a,b)}$, which we will describe in the following.

For (1,4)-semi-simple splicing, for each marker $(a,b) \in M$, we define the bridge $B_{(a,b)}$ as an automaton with an initial state $i_{(a,b)}$, an intermediate state $p_{(a,b)}$, a final state $t_{(a,b)}$ and two transitions $i_{(a,b)} \xrightarrow{a} p_{(a,b)}$ and $p_{(a,b)} \xrightarrow{b} t_{(a,b)}$. The bridge $B_{(a,b)}$ for (1,4)-semi-simple splicing rules is shown in Figure~\ref{fig:bridge14}.

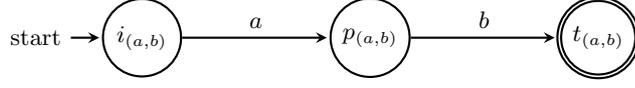
\begin{figure}
\begin{center}
\begin{tikzpicture}[shorten >=1pt,on grid,node distance=3cm,>=stealth,thick,auto]
    \node[state,initial] (0) {$i_{(a,b)}$};
    \node[state] (1) [right of=0]{$p_{(a,b)}$};
    \node[state,accepting] (2) [right of=1]{$t_{(a,b)}$};
    \path[->]
        (0) edge  node {$a$} (1)
        (1) edge  node {$b$} (2)
        ;
\end{tikzpicture}
\caption{The bridge $B_{(a,b)}$ for (1,4)-semi-simple splicing rules}
\label{fig:bridge14}
\end{center}
\end{figure}

As with the construction of the NFA for (2,4)-semi-simple splicing systems in Proposition~\ref{prop:pixton-construction}, we add transitions iteratively. Since the additional transitions depend on paths of length one, we can accomplish our construction with a maximum of two iterations. We begin with
\begin{equation*}
\delta_0' = \delta \cup \bigcup_{(a,b) \in M} \left\{ i_{(a,b)} \xrightarrow{a} p_{(a,b)}, p_{(a,b)} \xrightarrow{b} t_{(a,b)}  \right\}.
\end{equation*}
For the first iteration of the construction, we add $\varepsilon$-transitions between states of the original NFA and states belonging to the bridges. Then after the first iteration of additional transitions, we have the transition function 
\begin{equation*}
\delta_1'= \delta_0' \cup \bigcup_{(a,b) \in M} (\{q \xrightarrow{\varepsilon} i_{(a,b)} \mid \text{$\delta(q,a)$ is useful}\} \cup \{t_{a,b} \xrightarrow{\varepsilon} \im (\delta_0')_b\}).
\end{equation*}
For additional iterations, $\varepsilon$-transitions may be added between states belonging to the bridges in one of four ways, for $(a,b), (a',b') \in M$:
\begin{enumerate}
\item $i_{(a',b')} \xrightarrow{\varepsilon} i_{(a,b)}$ for $a' = a $,
\item $p_{(a',b')} \xrightarrow{\varepsilon} i_{(a,b)}$ for $b' = a$,
\item $t_{a,b} \xrightarrow{\varepsilon} t_{(a',b')}$ for $b' = b$,
\item $t_{a,b} \xrightarrow{\varepsilon} t_{(a',b')}$ for $b' = b$.
\end{enumerate}
Since there are finitely many transitions of this form that can be added, there may be only finitely many additional iterations. Therefore we have $\delta' = \delta_k'$ for some finite $k$.

We can then remove $\varepsilon$-transitions, merging states $i_{(a,b)}$ and all its incoming and outgoing transitions into transitions on $a$, while all states $t_{(a,b)}$ and its incoming and outgoing transitions are merged and replaced with transitions on $b$. Furthermore, all states $p_{(a,b)}$ and $p_{(a',b')}$ are merged for $a = a'$.

Thus, after $\varepsilon$-transition removal, we obtain an NFA $A_H' = (Q'', \Sigma, \delta'', q_0, F)$, where $Q'' = Q \cup Q_M$ with $Q_M = \{p_a \mid (a,b) \in M \}$ and the transition function $\delta''$ is defined for $q \in Q''$ by
\begin{itemize}
\item $\delta''(q,a) = \{\delta(q,a)\}$ if $\delta(q,a)$ is useless or $(a,b) \not \in M$ for any $b \in \Sigma$,
\item $\delta''(q,a) = \{\delta(q,a)\} \cup \{p_a\}$ if $\delta(q,a)$ is useful and there exists $b \in \Sigma$ such that $(a,b) \in M$,
\item $\delta''(p_a,b) = \im \delta_b \cup \{p_b\}$ if $(a,b) \in M$,
\end{itemize}
and all other transitions are undefined.

Then the NFA $A_H'$ behaves as follows. Upon reading a symbol $a$ with $(a,b) \in M$ for some $b \in \Sigma$, there is a transition to a state $p_a$ for each $a$ with $(a,b) \in M$. From each state $p_a$, there are transitions on $b$ to each state in $\im \delta_b$ and $p_b$. Thus the NFA $A_H'$ accepts the language $L(H)$ and since the state set of $A_H'$ is $Q'' = Q \cup Q_M$, $A_H'$ has $n+m$ states.
\end{proof}

This construction immediately gives an upper bound of $2^{n+m}$ states necessary for an equivalent DFA via the subset construction, where $m$ is the number of symbols on the left side of each pair of rules in $M$. However, we will show via the following DFA construction that the upper bound is much lower than this.

\begin{proposition} \label{prop:14-upper}
Let $H = (\Sigma,M,L)$ be a (1,4)-semi-simple splicing system with a regular initial language, where $M = M_1 \times M_2$ with $M_1, M_2 \subseteq \Sigma$ and $L \subseteq \Sigma^*$ is recognized by a DFA with $n$ states. Then there exists a DFA $A_H$ such that $L(A_H) = L(H)$ and $A_H$ has at most $(2^n - 2)(|M_1|+1) + 1$ states.
\end{proposition}

\begin{proof}
Let $A = (Q,\Sigma,\delta,q_0,F)$ be a DFA for $L$. We will define the DFA $A_H = (Q',\Sigma,\delta',q_0',F')$. Then the state set of $A_H$ is $Q' = 2^Q \times (M_1 \cup \{\varepsilon\})$, the initial state is $q_0' = \langle \{q_0\},\varepsilon \rangle$, the set of final states is $F' = \{\langle P,a \rangle \mid P \cap F \neq \emptyset \}$, and the transition function $\delta'$ is defined 
\begin{itemize}
\item $\delta'(\langle P,\varepsilon \rangle, a) = \langle P', \varepsilon \rangle$ if $a \not \in M_1$,
\item $\delta'(\langle P,\varepsilon \rangle, a) = \langle P', a \rangle$ if $a \in M_1$,
\item $\delta'(\langle P,b \rangle, a) = \langle P', \varepsilon \rangle$ if $(b,a) \not \in M$ and $a \not \in M_1$,
\item $\delta'(\langle P,b \rangle, a) = \langle P', a \rangle$ if $(b,a) \not \in M$ and $a \in M_1$,
\item $\delta'(\langle P,b \rangle, a) = \langle \im \delta_a, \varepsilon \rangle$ if $(b,a) \in M$ and $a \not \in M_1$,
\item $\delta'(\langle P,b \rangle, a) = \langle \im \delta_a, a \rangle$ if $(b,a) \in M$ and $a \in M_1$,
\end{itemize}
where $P' = \bigcup_{q \in P} \delta(q,a)$. 

This construction gives an immediate upper bound of $(2^n - 1)(|M_1|+1)$ states, however, not all of these states are distinguishable. Consider the two states $\langle Q,\varepsilon \rangle$ and $\langle Q,a \rangle$ for some $a \in M_1$. We claim that these two states are indistinguishable. This arises from the observation that $\bigcup_{q \in Q} \delta(q,a) = \im \delta_a$ for all $a \in \Sigma$. Then one of the following occurs:
\begin{itemize}
\item $\langle Q,\varepsilon \rangle \xrightarrow{b} \langle \im \delta_b,\varepsilon \rangle$ and $\langle Q,a \rangle \xrightarrow{b} \langle \im \delta_b,\varepsilon \rangle$ if $b \not \in M_1$,
\item $\langle Q,\varepsilon \rangle \xrightarrow{b} \langle \im \delta_b,b \rangle$ and $\langle Q,a \rangle \xrightarrow{b} \langle \im \delta_b,b \rangle$ if $b \in M_1$.
\end{itemize}
Note that in either case, it does not matter whether or not $(a,b) \in M$ and the two cases are distinguished solely by whether or not $b$ is in $M_1$. Thus, all states $\langle Q,a \rangle$ with $a \in M_1 \cup \{\varepsilon\}$ are indistinguishable. 

Thus, $A_H$ has at most $(2^n-2)(|M_1|+1)+1$ states.
\qed
\end{proof}

When the initial language is a regular language, the upper bound is easily reached, even when we are restricted to simple splicing rules.
\begin{lemma} \label{lem:14-lower}
There exists a (1,4)-simple splicing system with a regular initial language $H = (\Sigma,M,L_n)$ with $|M| = 1$ where $L_n$ is a regular language with state complexity $n$ such that the minimal DFA for $L(H)$ requires at least $(2^n-2)(|M_1|+1)+1$ states.
\end{lemma}

\begin{proof}
Let $A_n = (Q_n,\Sigma,\delta_n,0,F_n)$ be the DFA that recognizes $L_n$ with $Q_n = \{0,\dots,n-1\}$, $\Sigma = \{a,b,c\}$,  $F_n = \{0\}$ and the transition function is defined by
\begin{itemize}
\item $\delta(i,a) = i+1 \bmod n$ for all $0 \leq i \leq n-1$,
\item $\delta(i,b) = i$ for $0 \leq i \leq n-2$, $\delta(n-1,b) = 0$,
\item $\delta(i,c) = i$ for $0 \leq i \leq n-1$.
\end{itemize}
The DFA $A_n$ is shown in Figure~\ref{fig:14-witness}.

\begin{figure}
\begin{center}
\begin{tikzpicture}[shorten >=1pt,on grid,node distance=1.6cm,>=stealth,thick,auto]
    \node[state,initial,accepting] (0) {$0$};
    \node[state] (1) [right of=0]{$1$};
    \node[state] (2) [right of=1]{$2$};
    \node (3) [right of=2]{$\cdots$};
    \node[state] (4) [right of=3]{${{n-2}}$};
    \node[state] (5) [right of=4]{${{n-1}}$};
    \path[->]
        (0) edge [loop above] node {$b,c$} (0)
        (1) edge [loop above] node {$b,c$} (1)
        (2) edge [loop above] node {$b,c$} (2)
        (4) edge [loop above] node {$b,c$} (4)
        (5) edge [loop above] node {$c$} (5)
        (0) edge  node {$a$} (1)
        (1) edge  node {$a$} (2)
        (2) edge  node {$a$} (3)
        (3) edge  node {$a$} (4)
        (4) edge  node {$a$} (5)
        (5) edge [bend left=25] node {$a,b$} (0)
        ;
\end{tikzpicture}
\caption{The DFA $A_n$ for Lemma~\ref{lem:14-lower}}
\label{fig:14-witness}
\end{center}
\end{figure}
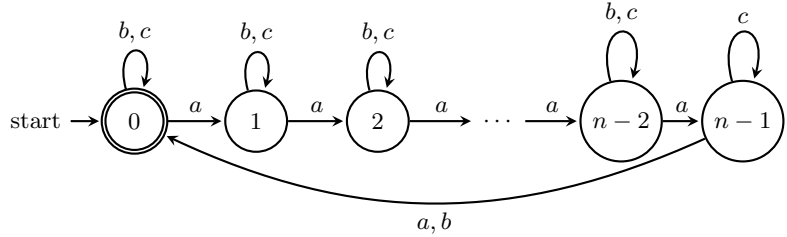

We consider the (1,4)-simple splicing system with a regular initial language $H = (\Sigma,\{(c,c)\},L_n)$ and consider the DFA $A_n'$ obtained via the construction of Proposition~\ref{prop:14-upper}. States of $A_n'$ are of the form $\langle P,\sigma \rangle$ for $P \subseteq Q_n$ and $\sigma \in M_1 \cup \{\varepsilon\}$. Note that $P \neq \emptyset$ since the empty set is not reachable. We will show that all such states with non-empty $P$ are reachable and pairwise distinguishable.

First, observe that $\langle 0,\varepsilon \rangle \xrightarrow{c^2} \langle Q_n,c \rangle \xrightarrow{a} \langle Q_n,\varepsilon \rangle$. Then we will show that for all nonempty subsets $S \subseteq Q_n$, every state $\langle S,\varepsilon \rangle$ is reachable by showing that it can be reached from $\langle Q_n, \varepsilon \rangle$. We have already shown that for the sole subset of $Q_n$ of size $n$, $Q_n$, the state $\langle Q_n,\varepsilon \rangle$ is reachable from the initial state. 

Next, we will show that we can reach a state $\langle S,\varepsilon \rangle$ where $S$ is a subset of size $k-1$ from some state $\langle T, \varepsilon \rangle$, where $T$ is a subset of size $k \geq 2$. Suppose that we can reach $\langle S,\varepsilon \rangle$ for a subset $S \subseteq Q_n$ of size $k$ and we wish to reach the state $\langle S \setminus \{t\}, \varepsilon \rangle$ for some $t \in Q_n$. There are two cases.

If $t+1 \in S$, then we have 
\begin{equation*}
\langle S,\varepsilon \rangle \xrightarrow{a^{n-1-t} b a^{t+1}} \langle S \setminus \{t\}, \varepsilon \rangle.
\end{equation*}
The same argument holds for $t = n-1$ and $0 \in S$.

On the other hand, if $t+1 \not \in S$, then we must first reach state $\langle S',\varepsilon \rangle$, where $S' = \delta'(\langle S,\varepsilon \rangle,a^{n-1-t})$. Observe that $t \xrightarrow{a^{n-1-t}} n-1$ and thus $n-1 \in S'$. From $\langle S', \varepsilon \rangle$, we want to reach the state $\langle S' \setminus \{n-1\}, \varepsilon \rangle$. Let $s = \min S'$. Then 
\begin{equation*}
\langle S', \varepsilon \rangle \xrightarrow{b} \langle S' \setminus \{n-1\} \cup \{0\}, \varepsilon \rangle \xrightarrow{(a^{n-1}b)^s a^s} \langle S' \setminus \{n-1\}, \varepsilon \rangle.
\end{equation*}
Finally, we shift every element of $S'$ back to its original position in $S$ by
\begin{equation*}
\langle S' \setminus \{n-1\}, \varepsilon \rangle \xrightarrow{a^{t+1}} \langle S \setminus \{t\}, \varepsilon \rangle,
\end{equation*}
and we have reached $\langle S \setminus \{t\}, \varepsilon \rangle$ as desired. Thus, we have shown that we can reach a state $\langle S,\varepsilon \rangle$ where $S$ is a subset of $Q_n$ of size $k-1$ from a state $\langle T,\varepsilon \rangle$ with a subset $T$ of $Q_n$ of size $k$.

Then, from each state $\langle S,\varepsilon \rangle$, the state $\langle S,c \rangle$ is reachable via the word $c$. Thus, every state of the form $\langle S,a \rangle$ for nonempty $S \subseteq Q_n$ and $a \in M_1 \cup \{\varepsilon\}$ is reachable.

To show that each of these states is pairwise disjoint, consider two states $\langle P,a \rangle$ and $\langle P',a' \rangle$, with nonempty $P,P' \subseteq Q_n$ and $a,a' \in M_1 \cup \{\varepsilon\}$. First, suppose that $P \neq P'$. Then there exists an element $t \in P$ such that $t \not \in P'$. Then $\langle P,a \rangle$ and $\langle P',a' \rangle$ are distinguishable via the word $a^{n-1-t}$. 

Now suppose that $P = P' \neq Q_n$ and $a \neq a'$. Thus, we consider two states $\langle P,\varepsilon \rangle$ and $\langle P,c \rangle$. We have $\langle P,\varepsilon \rangle \xrightarrow{c} \langle P,c \rangle$ and $\langle P,c \rangle \xrightarrow{c} \langle Q_n,c \rangle$. Since $P \neq Q_n$, these two states fall under the above case and are distinguishable. Finally, if $P = Q_n$, then the two states are not distinguishable, as shown in the proof of Proposition~\ref{prop:14-upper}.  

Thus, we have shown that $A_n'$ contains $(2^n - 2)(|M_1|+1) + 1$ reachable and pairwise distinguishable states.
\qed
\end{proof}

We note that the witness of Lemma~\ref{lem:14-lower} had $|M| = 1$ and therefore $|M_1| = 1$. It is not difficult to see that we can set $|M_1|$ to be arbitrarily large by adding symbols with transitions that behave the same way as $c$ and adding the corresponding markers to $M$ for each new such symbol.

By Proposition~\ref{prop:14-upper} and Lemma~\ref{lem:14-lower} we have the following result.
\begin{theorem}
Let $H = (\Sigma,M,L)$ be a (1,4)-semi-simple splicing system with a regular initial language, where $L \subseteq \Sigma^*$ is a regular language with state complexity $n$ and $M = M_1 \times M_2$ with $M_1,M_2 \subseteq \Sigma$. Then the state complexity of $L(H)$ is at most $(2^n - 2)(|M_1|+1) + 1$ and this bound can be reached in the worst case.
\end{theorem}

We will show that this bound cannot be reached by any (1,4)-semi-simple splicing system when the initial language is finite. 
\begin{proposition} \label{prop:14-semi-finite-upper}
Let $H = (\Sigma,M,I)$ be a (1,4)-semi-simple splicing system with a finite initial language, where $M = M_1 \times M_2$ with $M_1, M_2 \subseteq \Sigma$ and $I \subseteq \Sigma^*$ is a finite language recognized by a DFA with $n$ states. Then there exists a DFA $A_H$ such that $L(A_H) = L(H)$ and $A_H$ has at most $2^{n-2} +|M_1| \cdot 2^{n-3} + 1$ states. 
\end{proposition}

\begin{proof}
Let $A = (Q,\Sigma,\delta,q_0,F)$ be a DFA for $I$ with $n$ states and let $A_H$ be the DFA recognizing $L(H)$ obtained via the construction of Proposition~\ref{prop:14-upper}. Since $I$ is finite, the initial state of $A$ contains no incoming transitions and $A$ must have a sink state. Therefore, for any state $\langle S,c \rangle$, we have $S \subseteq Q \setminus \{q_0,q_\emptyset\}$ and $c \in M_1 \cup \{\varepsilon\}$, where $q_\emptyset$ is the sink state. This gives us up to $(2^{n-2} - 1) (|M_1| + 1) + 2$ states. 

We can reduce the number of reachable states further by noting that since $I$ is finite, $A$ must contain at least one useful state $q_1$ that is directly reachable only from the initial state $q_0$. Then there are only two ways to reach a state $\langle P,c \rangle$ in $A_H$ with $q_1 \in P$. Either $P = \{q_1\}$ and is reached directly via a transition from $\{q_0\}$ or $|P| \geq 2$ and $P = \im \delta_b$ for some $(a,b) \in M$. For each $c \in M_1$, this gives a total of 2 reachable states $\langle P,c \rangle$.

Therefore, we can enumerate the reachable states of $A_H$ as follows:
\begin{itemize}
\item the initial state $\langle \{q_0\},\varepsilon \rangle$,
\item the sink state $\langle \{q_\emptyset\},\varepsilon \rangle$,
\item at most $2^{n-2} - 1$ states of the form $\langle P,\varepsilon \rangle$, where $P \subseteq Q \setminus \{q_0,q_\emptyset\}$,
\item at most $|M_1|$ states of the form  $\langle \{q_1\},c \rangle$ with $c \in M_1$,
\item at most $|M_1|$ states of the form  $\langle P,c \rangle$ such that $P \subseteq Q \setminus \{q_0,q_\emptyset\}$, $|P| \geq 2$, and $q_1 \in P$ with $c \in M_1$,
\item at most $|M_1| (2^{n-3} - 1)$ states of the form  $\langle P,c \rangle$ such that $P \subseteq Q \setminus \{q_0,q_1,q_\emptyset\}$ with $c \in M_1$.
\end{itemize}
This gives a total of at most $2^{n-2} + |M_1| \cdot (2^{n-3}+1) + 1$ reachable states in~$A_H$.
\qed
\end{proof}

\begin{lemma} \label{lem:14-semi-finite-lower}
There exists a (1,4)-semi-simple splicing system with a finite initial language $H = (\Sigma,M,I_n)$, where $M = M_1 \times M_2$ with $M_1, M_2 \subseteq \Sigma$ and $M_1 \cap M_2 = \emptyset$ and $I_n$ is a finite language with state complexity $n$ such that a DFA recognizing $L(H)$ requries $2^{n-2} + |M_1| \cdot (2^{n-3}+1) + 1$ states.
\end{lemma}

\begin{proof}
We will consider the following family of splicing systems. Let $A_n = (Q_n,\Sigma_n,\delta_n,0,\{n-2\})$ be the DFA recognizing $I_n$, where $Q_n = \{0, \dots, n-1\}$ and $\Sigma = \{b,c,d\} \cup \Gamma_n$ where $\Gamma_n = \{a_S \mid S \subseteq \{1, \dots, n-2\} \}$. We define the transition function $\delta_n$ by
\begin{itemize}
\item $\delta_n(i,a_S) = \min\{j \in S \mid i < j \leq n-2\}$ for $0 \leq i \leq n-2$,
\item $\delta_n(i,b) = i+1$ for $0 \leq i \leq n-2$,
\item $\delta_n(i,c) = i+1$ for $0 \leq i \leq n-2$,
\item $\delta_n(i,d) = i+1$ for $0 \leq i \leq n-2$,
\item and all other transitions are to $n-1$.
\end{itemize} 
Let $M_n = \{b\} \times \Gamma_n \cup \{(b,d),(d,b)\}$. We consider the (1,4)-semi-simple splicing system $H = \{\Sigma_n, M_n, I_n\}$. Let $A_n'$ be the NFA recognizing $L(H)$ obtained via the construction from Proposition~\ref{prop:14-upper} and consider the DFA that results from applying the subset construction. 

Let us consider the number of reachable states of $2^{Q_n} \times \{b,d,\varepsilon\}$. First, the initial state $\langle \{0\},\varepsilon \rangle$ is reachable by definition and the sink state $\langle \{n-1\},\varepsilon \rangle$ is reachable on the word $c^{n-1}$. 

Now we consider states $q = \langle S,\varepsilon \rangle$, where $S \subseteq \{1,\dots,n-2\}$. From the initial state, we can reach states of the form $\langle T,\varepsilon \rangle$ with $T \subseteq \{1, \dots, n-2\}$ via the word $b a_T$ and there are $2^{n-2} - 1$ such states. 

Next, we consider states $q = \langle S,b \rangle$ where $S \subseteq \{1,\dots,n-2\}$. By Proposition~\ref{prop:14-semi-finite-upper}, there are exactly two states $\langle S',b \rangle$ that are reachable with $1 \in S'$. Either $S' = \{1\}$ or $S' = \im (\delta_n)_b$. The state $\langle \{1\},b \rangle$ is reachable from the initial state via the word $b$ while the state $\langle \im (\delta_n)_b, b \rangle$ is reachable from the initial state via the word $db$. 

Now consider a subset $R = \{i_1, \dots, i_k\}$ with $2 \leq i_1 < \cdots i_k \leq n-2$. To reach the state $\langle R,b \rangle$, let $R' = \{i_1 - 1, \dots, i_k - 1\}$. Since $R' \subseteq \{2, \dots, n-2\}$, we have
\begin{equation*}
\langle \{0\},\varepsilon \rangle \xrightarrow{b a_{R'}} \langle R',\varepsilon \rangle \xrightarrow{b} \langle R,b \rangle.
\end{equation*}  
There are $2^{n-3} - 1$ such states, giving a total of $2^{n-3}$ reachable states of the form $\langle S,b \rangle$. A similar argument holds for states of the form $\langle S,d \rangle$ with $S \subseteq \{1,\dots,n-2\}$.

To show that each of these states is pairwise distinguishable, consider two states $\langle S,\sigma \rangle$ and $\langle S', \sigma' \rangle$ for $S,S' \subseteq Q_n$ and $\sigma, \sigma' \in \{b,d,\varepsilon\}$. First, suppose that $S \neq S'$. Then without loss of generality, there exists an element $t \in S$ that is not in $S'$ and the two states can be distinguished by the word $c^{n-2-t}$. 

Now, suppose $S = S'$ and $\sigma \neq \sigma'$. First, consider when $\sigma \in \{b,d\}$ and $\sigma' = \varepsilon$. Let $S = \{i_1, \dots, i_k\}$ and $\sigma = d$. Then $\langle S,\sigma \rangle \xrightarrow{b} \langle \{1, \dots, n-2\}, \varepsilon \rangle$ and $\langle S',\varepsilon \rangle \xrightarrow{b} \langle \{i_1 + 1, \dots, i_k + 1\}, b \rangle$ and the two resultant states can be distinguished as in the case above. The argument is similar for $\sigma = b$

Next, suppose $\sigma = b$ and $\sigma' = d$. We have $\langle S,b \rangle \xrightarrow{d} \langle \{1,\dots,n-2\},d \rangle$ and $\langle S,d \rangle \xrightarrow{d} \langle T,d \rangle$, where $T = \bigcup_{q \in S} \delta_n(q,d)$. Since $I_n$ is a finite language, we know that $1 \not \in T$ and therefore, $T \neq \{1,\dots,n-2\}$ and the two states can be distinguished as above. Again, the argument is similar with $\sigma = d$ and $\sigma' = b$. Thus, all states $\langle S,\sigma \rangle$ and $\langle S',\sigma' \rangle$ are distinguishable.

Therefore, $A_n'$ has $2^{n-2} + |M_1| \cdot (2^{n-3}+1) + 1$ reachable and pairwise distinguishable states.
\qed
\end{proof}
We note that one can arbitrarily increase the size of $M$ by adding symbols $s$ and $t$ with the same role as $b$ and $d$, respectively, and the corresponding pairs of rules $(s,t)$ and $(t,s)$.

Together, Proposition~\ref{prop:14-semi-finite-upper} and Lemma~\ref{lem:14-semi-finite-lower} give the following result.

\begin{theorem}
Let $H = (\Sigma,M,I)$ be a (1,4)-semi-simple splicing system with a finite initial language, where $I \subseteq \Sigma^*$ is a finite language with state complexity $n$ and $M = M_1 \times M_2$ with $M_1, M_2 \subseteq \Sigma$. Then the state complexity of $L(H)$ is at most $2^{n-2} + |M_1| \cdot 2^{n-3} + 1$ and this bound is reachable in the worst case.
\end{theorem}

To conclude this section, we observe, as noted in Section~\ref{sec:24}, that when the visible site is 4, as in (1,4)- and (2,4)- splicing systems, lower bound witnesses with finite initial languages must be defined over an alphabet that grows exponentially with the number of states in order to reach the upper bound. This is in contrast to (2,3)-semi-simple splicing systems, as in Section~\ref{sec:23} and (1,3)-semi-simple splicing systems, studied in \cite{Kari2019a}. In both of these cases, lower bound witnesses with a fixed size alphabet sufficed.

\section{Conclusion} \label{sec:concl}
We have studied the state complexity of several variants of semi-simple splicing systems. Our results are summarized in Table~\ref{tab:conclusion} and we include the state complexity of (1,3)-semi-simple and (1,3)-simple splicing systems from \cite{Kari2019a} for comparison. We observe that for all variants of semi-simple splicing systems, the state complexity bounds for splicing systems with regular initial languages are reached with simple splicing witnesses defined over a three-letter alphabet.

{\setlength{\tabcolsep}{6pt}
\renewcommand{\arraystretch}{1.25}
\begin{table}
    \centering
    \begin{tabular}{@{}|l|l|l|@{}}
    \toprule
     & \textbf{Regular axiom set} & \textbf{Finite axiom set} \\
    \midrule
    (2,4)-semi. & $2^n-1, |\Sigma| = 3$ & $2^{n-2} + 1, |\Sigma| \geq 2^{n-3}$ \\
    (2,3)-semi. & $2^{n-1}, |\Sigma| = 3$& $2^{n-3}+2, |\Sigma| = 3$ \\
    (1,4)-semi. & $(2^{n-2} - 2)(|M_1| + 1) + 1, |\Sigma| = 3$& $2^{n-2} + |M_1|\cdot 2^{n-3}, |\Sigma| \geq 2^{n-3}$ \\
    \emph{(1,3)-semi.} \cite{Kari2019a} & $2^n-1, |\Sigma| = 3$& $2^{n-2} + 1, |\Sigma| = 3$ \\
    \midrule
    (2,4)-simple & $2^n-1, |\Sigma| = 3$ & \emph{Same as (1,3)}\\
    (2,3)-simple & $2^{n-1}, |\Sigma| = 3$& $2^{n-4} + 2^{n-5} + 2, |\Sigma| = 7$\\
    (1,4)-simple & $(2^{n-2} - 2)(|M_1| + 1) + 1, |\Sigma| = 3$& \emph{?} \\
    \emph{(1,3)-simple} \cite{Kari2019a} & $2^n-1, |\Sigma| = 3$& $2^{n-2} + 1, |\Sigma| \geq 2^{n-3}$\\
    \bottomrule
    \end{tabular}
    \vspace*{2mm}
    \caption{Summary of state complexity bounds for $(i, j)$ simple splicing systems and semi-simple splicing systems with alphabet $\Sigma$, state complexity of the axiom  $n$, and set of splicing rules $M = M_1 \times M_2$, with $M_1, M_2 \subseteq \Sigma$.}
    \label{tab:conclusion}
\end{table}
}

For semi-simple splicing systems with finite initial languages, we observe that the state complexity bounds for the (2,3) and (1,3) variants are reached by witnesses defined over a three-letter alphabet, while both of the (1,4) and (2,4) variants require an alphabet size that is exponential in the size of the DFA for the initial language.

For simple splicing systems with finite initial languages, since (1,3)- and (2,4)-simple splicing systems are equivalent, the bound is reached by the same witness from \cite{Kari2019a}. The witness for (2,3)-simple splicing systems with a finite initial language is defined over a fixed alphabet of size 7, while the problem remains open for (1,4)-simple splicing systems.

Another problem that remains open is the state complexity of (1,4)- and (2,4)- simple and semi-simple splicing systems with finite initial languages defined over alphabets of size $k$ for $3 < k < 2^{n-3}$. A similar question can be asked of (2,3)-simple splicing systems with a finite initial language for alphabets of size less than 7.

\bibliographystyle{splncs_srt}
\bibliography{library}
\end{document}